\documentclass[9pt,shortpaper,twoside,web]{ieeecolor}
\usepackage{generic}

\usepackage{amsmath,amssymb} 
\newcommand {\matr}[2]{\left[\begin{array}{#1}#2\end{array}\right]}

\newcommand{\tightMatr}[2]{\begin{bmatrix}#2\end{bmatrix}}

\usepackage{algorithm}
\usepackage[noend]{algpseudocode}
\usepackage{graphicx}
\graphicspath{{images/}}

\usepackage{tabularx,booktabs}
\definecolor{wheat}{rgb}{0.96,0.87,0.70}

\usepackage{xcolor, soul}

\usepackage{accents}
\usepackage{expl3,xparse}

\newcommand{\E}{\mathbb{E}}

\newcommand{\rr}{{\mathbb R}}

\newcommand{\matrice}[2]{\left[\hspace*{-.1cm}\begin{array}{#1} #2 \end{array}\hspace*{-.1cm}\right]}

\newtheorem{Theorem}{Theorem}
\newtheorem{Lemma}[Theorem]{Lemma}
\newtheorem{Proposition}[Theorem]{Proposition}

\newtheorem{Problem}[Theorem]{Problem}

\newtheorem{Example}{Example}
\newtheorem{Remark}{Remark}

\usepackage{tikz}
\usepackage{cite}
\usepackage{soul}

\title{\LARGE \bf {Constrained Controller and Observer Design by Inverse Optimality}}

\author{Mario Zanon and Alberto Bemporad
\thanks{
M. Zanon and A. Bemporad are with IMT Lucca, Lucca, Italy. e-mail: \{name.surname@imtlucca.it\}.%
}%
}

\begin{document}

\maketitle


%
%
%
%
\begin{abstract}
	Model Predictive Control (MPC) is often tuned by trial and error.
When a baseline linear controller exists that is already well tuned in the absence of constraints and MPC is introduced to enforce them, one would like to 
avoid altering
        the original linear feedback law whenever they are not active. We formulate this problem as a controller matching similar to~\cite{DiCairano2009,DiCairano2010,Tran2015}, which we extend to a more general framework. We prove that a positive-definite stage cost matrix yielding this matching property can be computed for all stabilizing linear controllers. Additionally, we prove that the constrained estimation problem can also be solved similarly, by matching a linear observer with a Moving Horizon Estimator (MHE). Finally, we discuss various aspects of the practical implementation of the proposed technique in some examples.
\end{abstract}


\section{Introduction}\label{sec:intro}
	Model Predictive Control (MPC) provides a systematic approach to control systems
    subject to constraints, by relying on constrained optimization. Recent progress on solvers has made MPC applicable at unprecedented high rates, thus widely enlarging its range of potential applications. While constraint satisfaction is enforced automatically by the optimization procedure, closed-loop tracking performance is achieved by adequately tuning the cost function to be minimized.
	
	The main drawback of MPC is the lack of a systematic approach to tune its cost function. Some approaches have been developed to tune the controller for specific definitions of the control performance. In case a clear performance criterion can be formulated as a function of the states and controls (often referred to as economic MPC), an approach for computing a quadratic positive-definite stage cost was proposed in~\cite{Zanon2014d,Zanon2016b,Zanon2017a}. Moreover, (semi)-automatic tuning methods were proposed in~\cite{FPB20} and \cite{ZBP21} based on black-box global and preference-based optimization, respectively.
    In many cases, however, the standard procedure consists of adapting the MPC cost parameters by trial-and-error until closed-loop performance is satisfactory. 
	
	Since a plethora of tuning methods were developed for linear controllers, forcing MPC to match the feedback law of a well-designed linear controller whenever possible is of practical interest. Therefore, a controller matching procedure was proposed in~\cite{DiCairano2009,DiCairano2010,Tran2015}, with the objective of computing a stage cost for MPC that, whenever possible, delivers a feedback control coinciding with the one of a prescribed linear controller. The tuning procedure  consists of (a) tuning a linear controller using one of the many available methods, and (b) solving the controller matching problem to obtain a suitable cost function for MPC. 
	
	In~\cite{DiCairano2009,DiCairano2010}, a controller matching in state space was proposed, but the cost was restricted to have zero cross state-input terms, such that some controllers could only be matched approximately. In~\cite{Tran2015}, an input-output setting was considered and the norm of the difference between the MPC and the desired feedback matrix minimized. However, no guarantee that the feedback matrix can be recovered exactly was given. Additionally, it was left as an open question whether an indefinite cost can leave more freedom to match a wider range of controllers.
    
    In this paper we close a theoretical gap by proving that every stabilizing linear feedback controller can be matched exactly by a positive-definite stage cost in MPC. Additionally, we provide three different methods for solving the controller matching problem which are easy to implement. Our derivation is first done for models in state-space form and then extended to the input-output case. 
	Finally, we show that our developments also apply to state estimation, such that constraints on state estimates are handled by a moving horizon estimator that matches a prescribed linear observer.
	
	The paper is structured as follows. In Section~\ref{sec:formulation} we prove that 
    every stabilizing linear feedback controller can be matched exactly by a linear quadratic regulator (LQR). We propose three solution strategies based on solving a small-dimensional semidefinite programming (SDP) problem in Section~\ref{sec:solving_the_problem}. We comment on how to deploy our results for reference tracking both in case of state-space and input-output models in Section~\ref{sec:tracking_io}. In Section~\ref{sec:mpc} we prove that the controller-matching property proven for LQR holds for MPC as long as the constraints are not active. We briefly discuss the observer matching problem in Section~\ref{sec:observer}. Using four examples, we demonstrate the effectiveness of the matching procedure and discuss practical implementation aspects in Section~\ref{sec:simulations}. We finally draw conclusions in Section~\ref{sec:conclusions}.
	
\section{Problem Formulation}
\label{sec:formulation}
Consider the linear discrete-time system
\begin{equation}
	x_{+} = Ax + B u,
    \label{eq:system}
\end{equation}
where $x\in\rr^{n_x}$ is the state vector, $u\in\rr^{n_u}$ is the input vector, and
$x_+$ is the state at the next time. Assume that a \emph{linear} feedback law
\begin{equation}
    u= -\hat Kx
\label{eq:Kx}
\end{equation}
which asymptotically stabilizes~\eqref{eq:system} has been designed to yield the desired closed-loop performance. 
Our goal is to design a \emph{model predictive controller} that: (a) enforces the constraints
\begin{equation}
    Cx+Du+e\leq 0,
\label{eq:constraints}
\end{equation}
defined by matrices $C$, $D$ and vector $e$; and (b) delivers a feedback law which exactly coincides with the linear control law in~\eqref{eq:Kx}, when the constraints in the MPC optimization problem are not active.

In order to address such a goal, we first neglect constraints~\eqref{eq:constraints}
and focus on the LQR problem 
\begin{subequations}
	\begin{align}
	\min_u \ \ & \sum_{k=0}^\infty \ell( x_k, u_k ) \nonumber\\
	\mathrm{s.t.} \ \ & x_{k+1} = Ax_k + B u_k, && k = 0,1,\ldots,
    \label{eq:LQR}
	\end{align}
where the stage cost is
\begin{align}
	\label{eq:stage_cost}
	\hspace{-0.8em}\ell(x,u) = \matr{c}{x \\ u}^\top \matr{ll}{Q & S^\top \\ S & R} \matr{c}{x \\ u} = \matr{c}{x \\ u}^\top
 H \matr{c}{x \\ u},
\end{align}
\label{eq:LQR-pb}%
\end{subequations}
with 
$H=H^\top\in\rr^{(n_x+n_u)\times (n_x+n_u)}$. 
The solution of the LQR problem~\eqref{eq:LQR-pb}, if it exists, is the only stabilizing solution among all solutions of the Discrete Algebraic Riccati Equation (DARE)
\begin{subequations}
	\label{eq:dare}
	\begin{align}
		&P =  A^\top P A  +  Q - (S^\top + A^\top P B) K,\\
		&(R + B^\top P B)K = S + B^\top P A.
	\end{align}
\end{subequations}

\begin{Problem}[LQR controller matching]
	\label{prob:match}
	Given a linear model $(A,B)$ and an asymptotically stabilizing feedback matrix $\hat K$, design a positive-definite stage cost such that the corresponding LQR controller from~\eqref{eq:dare} is $K=\hat K$.
\end{Problem}

We focus on the discrete-time case, even though the same results also hold in continuous time, which is omitted for conciseness. 

For all positive-definite matrices $H$, the assumption that $(A,B)$ is stabilizable implies that the LQR feedback gain from~\eqref{eq:dare} is asymptotically stabilizing.
In case $H$ is not positive-definite, the additional asymptotic convergence constraint 
\begin{equation}
    \lim_{k\to\infty} x_k=0
\label{eq:as_convergence}
\end{equation}
is often necessary to guarantee that the solution is asymptotically stabilizing, as shown in the following example:
\begin{Example}[Indefinite LQR and DARE]
	Consider the scalar system $x_{k+1}=2x_k + u_k$ and stage cost $\ell(x_k,u_k)=u^2$. The 
    corresponding DARE is
	$P=4P-\frac{4P^2}{1+P}$,
	with solutions $P\in\{0,3\}$, $K\in\{0,1.5\}$. The first one is destabilizing and corresponds to the formulation without constraint~\eqref{eq:as_convergence}, the second one is stabilizing and corresponds to the constrained formulation.
\end{Example}

For more details on indefinite LQR formulations we refer the interested reader to~\cite{Willems1971,Zanon2014d,Molinari1975a,Gruene2018}.

\section{Solution to the Inverse LQR Problem}
\label{sec:solving_the_problem}

In order to discuss Problem~\ref{prob:match} we first establish some preliminary results. 
Let $A_{\hat K} := A - B \hat K$ and note that, for any matrix $\bar Q \succ 0$, asymptotic stability of $A_{\hat K}$ implies that the Lyapunov equation 
\begin{align}
	\label{eq:lyap}
	\bar Q + A_{\hat K}^\top \bar P A_{\hat K} - \bar P=0
\end{align}
is solved by some matrix $\bar P\succ 0$.

\begin{Lemma}
	\label{lem:zero_feedback}
	Consider the linear discrete-time system 
	\begin{align*}
	x_{+} = A_{\hat K}x + B u,
	\end{align*}
	with $A_{\hat K}$ asymptotically stable. Let $\bar P$ be the solution to the Lyapunov equation~\eqref{eq:lyap} for $\bar Q=\bar Q^\top\succ 0$, and select cost matrices $\bar Q$, $\bar S:=-B^\top \bar P A_{\hat K}$ and any $\bar R\succ0$. Then, the LQR feedback is $K=0$.
\end{Lemma}
\begin{proof}
	We begin by noting that $P=\bar P$, with $\bar P$ solving~\eqref{eq:lyap}, and $K=0$ solve the DARE~\eqref{eq:dare} associated with system $(A_{\hat K},B)$ and cost matrices $\bar Q$, $\bar S$,
for any $\bar R\succ0$. Since $A_{\hat K}$ has all eigenvalues inside the unit circle, $K=0$ stabilizes $(A_{\hat K},B)$. Then $P=\bar P$, $K=0$ is a stabilizing solution of the LQR. Since the stabilizing solution, when it exists, is unique~\cite{Molinari1975a,Zanon2014d} this concludes the proof.
\end{proof}
\begin{Lemma}[{\cite[Lemma~1]{Zanon2014d}}]
	\label{lem:discr_lqr_equivalence}
	Consider system $(A_{\hat K},B)$ with $A_{\hat K}$ asymptotically stable, cost matrices $\bar Q$, $\bar R$, $\bar S$ from Lemma~\ref{lem:zero_feedback}, and corresponding LQR feedback $\bar K=0$; and consider system $(A,B)$ with cost matrices $Q$, $R$, $S$ and corresponding LQR feedback $K$. Assume that
	\begin{align}
	Q &= \bar Q + \bar S^\top \hat K + \hat K^\top \bar S + \hat K^\top \bar R \hat K, \quad
	S = \bar S + R \hat K,\quad R = \bar R.
    \label{eq:QSR}
	\end{align}
	Then, starting from the same initial state, the two systems generate the same trajectories
    in closed-loop with the corresponding LQR law, where for system $(A,B)$ the LQR law is $K=\hat K$.
\end{Lemma}
\begin{proof}
	The proof given in~\cite{Zanon2014d} is obtained by noting that the DAREs associated with the two LQR formulations coincide.
\end{proof}

We are now ready to prove the following theorem.
\begin{Theorem}
	\label{thm:existence}
	Given a linear discrete-time stabilizable system $(A,B)$ and any asymptotically stabilizing feedback $\hat K$, there exists a quadratic positive-definite stage cost $\ell(x,u)$ as in~\eqref{eq:stage_cost} such that the corresponding LQR solution~\eqref{eq:dare} is $K=\hat K$.
\end{Theorem}

\begin{proof}
	The proof is based on using first Lemma~\ref{lem:zero_feedback} to construct a positive-definite LQR formulation for system $(A_{\hat K},B)$, with $A_{\hat K}= A - B \hat K$, and then prove that this implies the existence of a positive-definite LQR formulation also for system $(A,B)$.
	
	Select any matrix $\bar Q=\bar Q^\top \succ 0$, compute $\bar P$ by solving the Lyapunov equation~\eqref{eq:lyap}, and define $\bar S:=-B^\top \bar P A_{\hat K}$. 
	By selecting any symmetric matrix $\bar R$ such that $\bar R \succ \bar S \bar Q^{-1} \bar S^\top\succeq 0$, we get
	\begin{align*}
		\bar H := \matr{ll}{\bar Q & \bar S^\top \\ \bar S & \bar R } \succ 0.
	\end{align*}
	By Lemma~\ref{lem:zero_feedback}, this yields a positive-definite LQR formulation with zero feedback for system $(A_{\hat K},B)$, so that no control action is applied to system $(A_{\hat K},B)$.
	
	By applying Lemma~\ref{lem:discr_lqr_equivalence}, we obtain an equivalent LQR for system $(A,B)$ by defining the cost matrices $Q,R,S$ as in~\eqref{eq:QSR}.

We are left with proving that $H\succ 0$, or, equivalently, that $Q - S^\top R^{-1}S \succ 0$, since $R\succ0$. Because
	\begin{align*}
		S^\top R^{-1}S = \bar S^\top \bar R^{-1} \bar S + \bar S^\top \hat K + \hat K^\top \bar S + \hat K^\top \bar R \hat K,
	\end{align*}
we obtain
		$Q - S^\top R^{-1}S = \bar Q - \bar S^\top \bar R^{-1} \bar S \succ 0,$
	where positive-definiteness of the second term follows from $\bar H \succ 0$.
\end{proof}

By taking a different point of view, we provide next an alternative proof of Theorem~\ref{thm:existence}.

\begin{proof}[Alternative proof of Theorem~\ref{thm:existence}]
	For any $\Gamma \succ 0$ the cost%
	\begin{subequations}
		\label{eq:indef_cost}
		\begin{align}
			\ell(x,u) &= (u+\hat K x)^\top \Gamma(u+\hat K x) \\
			&= \matr{c}{x \\ u}^\top \matr{cc}{\hat K^\top \Gamma \hat K & \hat K^\top \Gamma \\ \Gamma \hat K & \Gamma} \matr{c}{x \\ u}
		\end{align}
	\end{subequations}
	solves the DARE~\eqref{eq:dare} with $K = \hat K$ and $P=0$. This fact is used in~\cite{Chisci2001} in the context of tube-based robust MPC. Since the proposed cost is indefinite, we exploit the results of~\cite{Zanon2014d,Zanon2016b} which state that any LQR with indefinite cost and stabilizing feedback matrix can be reformulated as an LQR with positive definite cost. 
\end{proof}
Since this theorem proves that any stabilizing feedback matrix $\hat K$ can be obtained as the solution of an LQR with positive-definite stage cost, there is no advantage in using an indefinite stage cost. Furthermore, 
we establish next a counterintuitive result about the solution of the controller matching problem for destabilizing feedback.
\begin{Lemma}
	Given any feedback $\hat K$ there exists a solution to the DARE formulated using cost~\eqref{eq:indef_cost} which yields $\hat K$ as feedback. Moreover, this entails that a stabilizing LQR solution exists, though $K=\hat K$ only holds if $\hat K$ is stabilizing.
\end{Lemma}
\begin{proof}
	The DARE reads
	\begin{align*}
		P &=  A^\top P A  +  \hat K^\top \Gamma \hat K - (\hat K^\top \Gamma + A^\top P B) \\
		&\hspace{10em}\cdot(\Gamma + B^\top P B)^{-1}(\Gamma \hat K + B^\top P A),
	\end{align*}
	such that $P=0$, $K=\hat K$ is a solution, though not necessarily a stabilizing one, of the DARE. The existence of a stabilizing solution of the indefinite LQR is then a direct consequence of this fact and~\cite[Proposition~2]{Zanon2014d}.
\end{proof}
This lemma warns the control engineer that the controller matching procedure might succeed at finding a positive-definite LQR formulation also in case of a non-stabilizing feedback $\hat K$; however, the LQR feedback is stabilizing, such that $K\neq \hat K$. We provide next a simple example to demonstrate this fact. 

\begin{Example}[Destabilizing Controller Matching]
	Consider the system $A = 0.9$, $B = 0.1$, with destabilizing feedback $\hat K = -2$.
	The indefinite LQR formulation using cost~\eqref{eq:indef_cost} with $\Gamma=1$, i.e., $Q = \hat K^\top \hat K$, $R=1$, $S=\hat K^\top$ yields the DARE
	\begin{align*}
	P = 0.81 P + 4 - (0.09P-2)^2(0.01P+1)^{-1},
	\end{align*}
	which simplifies to
	$-P\frac{P - 21}{P + 100} = 0$.
	This equation has two solutions: the stabilizing one corresponds to $P=21$ and $K=-0.0909$; and the destabilizing one  corresponds to $P=0$ and $K=-2$.
\end{Example}

We conclude this section by proving that in general there exist infinitely many LQR formulations yielding $K=\hat K$.
\begin{Proposition}
	\label{prop:inf_many_sol}
	Given a stabilizing feedback gain $\hat K$, there exist infinitely many LQR formulations yielding $K=\hat K$.
\end{Proposition}
\begin{proof}
	The LQR gain is invariant under the transformation~\cite{Gros2020}
	\begin{equation*}
		H \leftarrow H+\tightMatr{cc}{\,K^\top P_1 K & K^\top P_1 \, \\ P_1 K & P_1} + \tightMatr{cc}{ \, A^\top P_2 A - P_2 & A^\top P_2 B \, \, \\ B^\top P_2 A & B^\top P_2 B \, \,  },
	\end{equation*}
	for any $P_1$, $P_2$, provided that the following holds:
	\begin{align*}
		P_1 +B^\top P_2 B + R +  B^\top P B \succ 0.
	\end{align*}
In addition, the LQR gain is invariant under positive scaling, i.e., $H\leftarrow\sigma H$, for any $\sigma>0$.
\end{proof}

\vspace{0.5em}
\subsubsection*{Numerical Methods for the Inverse LQR Problem}
We propose two formulations based on semidefinite programming (SDP)~\cite{Boyd2004}: (i) a direct formulation which does not require any other information than $\hat K$; (ii) an indirect formulation based on cost~\eqref{eq:indef_cost}, where one needs to provide a tuning matrix $\Gamma$. We stress that the matching problem must be solved only once offline. Since there exist infinitely many cost matrices yielding an exact match, we minimize the condition number (the ratio between the maximum and minimum eigenvalue) of matrix $H$ in order to avoid numerical inaccuracies when later solving the MPC problem on line.
\paragraph*{Direct formulation}
Given the desired gain $\hat K$, solve 
\begin{subequations}
	\label{eq:sdp_tuning}
	\begin{align}
		\min_{Q,S,R,P,\beta} \ \ & \beta \\
		\mathrm{s.t.} \ \ & \beta I \succeq 
	\matr{ll}{Q&S^\top\\S&R} 
	\succeq I, \qquad  \beta I \succeq P \succeq I \\
		& P = A^\top P A + Q - (S^\top + A^\top P B) \hat K ,\\
		& (R + B^\top P B)\hat K=  S + B^\top P A.
	\end{align}
\end{subequations}
Let $H_\star= \left [\begin{smallmatrix}
	Q_\star&S_\star^\top\\S_\star&{}R_\star^{\phantom{\top}}
\end{smallmatrix}\right ]
$, $P_\star$, $\beta_\star$ be an optimal solution of~\eqref{eq:sdp_tuning}. 
The condition number $\kappa_\star$ of the stage cost matrix $H_\star$ clearly satisfies $\kappa_\star\leq \beta_\star$, which is the reason for minimizing $\beta$ in~\eqref{eq:sdp_tuning}. The lower bound $H\succeq I$ in~\eqref{eq:sdp_tuning} does not cause any loss of generality: since $\beta$ is not upper-bounded and scaling $H$ to $\sigma H$ does not change the optimizer for all $\sigma>0$, any $H\succ 0$ can be rescaled with $\sigma^{-1}=\lambda_{\rm min}(H)$ (the minimum eigenvalue of $H$), so that $\sigma H\succeq I$. The same reasoning holds for $P$. 

\paragraph*{Indirect formulation} 
For any given 
matrix $\Gamma=\Gamma^\top\succ0$, solve%
\begin{subequations}%
	\label{eq:sdp_tuning_gamma}%
	\begin{align}%
	\min_{P,\alpha,\beta} \ \ & \beta \\
	\mathrm{s.t.} \ \ & \beta I \succeq \alpha H_\Gamma + H_P \succeq I, \qquad  \beta I \succeq P \succeq I,
	\end{align}
\end{subequations}
where in~\eqref{eq:sdp_tuning_gamma} we have set
\begin{align*}
H_\Gamma &:= \tightMatr{cc}{\hat K^\top \Gamma \hat K & \hat K^\top \Gamma \\ \Gamma \hat K & \Gamma}\hspace{-1pt}, \ \
H_P := -\tightMatr{cc}{ A^\top P A - P & A^\top P B \\ B^\top P A & B^\top P B }\hspace{-1pt}.
\end{align*}
Here, we exploited the fact that, by Proposition~\ref{prop:inf_many_sol}, $H_P$ does not change the LQR solution and stage cost matrix $H_\Gamma$ yields feedback $K=\hat K$ by construction (see the alternative proof of Theorem~\ref{thm:existence}).
From SDP~\eqref{eq:sdp_tuning_gamma} we obtain $H= H_\Gamma + H_P $,
or, equivalently,%
\begin{subequations}
	\label{eq:qrs_gamma}
	\begin{align}
		Q&=\hat K^\top \Gamma \hat K+P-A^\top P A,\\
		R&=\Gamma-B^\top P B,\\
		S&=\Gamma \hat K-B^\top P A.
	\end{align}
\end{subequations}
Note that, as proven in~\cite{Zanon2014d}, $P$ is the cost-to-go matrix associated with 
stage-cost matrix $H$. Therefore, the considerations made for Problem~\eqref{eq:sdp_tuning} regarding the condition number of $H$ and $P$ directly apply to Problem~\eqref{eq:sdp_tuning_gamma}.

This second formulation allows one to tune the behavior in case some constraint becomes active, as one can adjust the way the MPC control deviates from the prescribed controller. This fact will be illustrated by an example in Section~\ref{ex:constrained}. Since it might not be clear how to best select $\Gamma$, one can let the optimizer select it by solving 
\begin{subequations}%
	\label{eq:sdp_tuning_gamma_opt}%
	\begin{align}%
		\min_{\Gamma,P,\beta} \ \ & \beta \\
		\mathrm{s.t.} \ \ & \beta I \succeq H_\Gamma + H_P \succeq I, \qquad  \beta I \succeq P \succeq I,
	\end{align}
\end{subequations}
where variable $\alpha$ has been removed, since $\Gamma$ is now an optimization variable.
This third formulation might be desirable when there is no clear criterion on how to select $\Gamma$ and the only objective is to obtain a numerically well conditioned cost matrix. Note that the solution of Problem~\eqref{eq:sdp_tuning_gamma_opt} coincides with that of Problem~\eqref{eq:sdp_tuning}, since both problems are convex and minimize the same cost.

\begin{Remark}
	Though the three formulations~\eqref{eq:sdp_tuning},~\eqref{eq:sdp_tuning_gamma}, and~\eqref{eq:sdp_tuning_gamma_opt} are all equivalent (see the alternative proof of Theorem~\ref{thm:existence}), in practice~\eqref{eq:sdp_tuning_gamma}  and~\eqref{eq:sdp_tuning_gamma_opt} were always solved by all the SDP solvers we tested, i.e., SeDuMi~\cite{Sturm1999}, SDPT3~\cite{Tutunku2003}, and Mosek~\cite{MOSEK-hp}. Problem~\eqref{eq:sdp_tuning}, instead, was harder to solve and in some cases the solvers were unable to compute a solution. 
\end{Remark}

\begin{Remark}
	When solving MPC problems, one can either keep the QP in a sparse form or condense it. The Hessian of the sparse QP is block diagonal, with $N$ blocks equal to $H$ and the last block equal to $P$. Therefore, the condition number of the sparse QP Hessian is the condition number of $\mathrm{blockdiag}(H,P)$. 
	The condensed QP Hessian is instead dense, since the states are eliminated using the solution formula
		$x_k = A^k x_0 + \sum_{j=0}^{k} A^j B u_j$.
	Because matrices $A$, $B$ are fixed, the condensed Hessian is a linear function of $H$ and $P$. One can therefore in principle minimize the condition number of the condensed QP. Note, however, that the condensed QP Hessian might be ill-conditioned and, therefore, pose difficulties to the SDP solver.
\end{Remark}

\section{Tracking Problems and Input-Output Form}
\label{sec:tracking_io}

In this section, we show how the results of the previous sections can be adapted to solve output tracking problems, both for state-space and input-output models.

\vspace{0.5em}
\subsubsection*{Tracking in State-Space Form}
Let $y\in\rr^{n_y}$ be the output vector associated with system~\eqref{eq:system},
\begin{equation}
	y = C_yx + D_yu.
    \label{eq:system-y}
\end{equation}
In order to achieve perfect tracking, as also suggested in~\cite{Kwakernaak1972},
one can design a linear controller for the extended system
\begin{equation}
    \matrice{c}{x_{k+1}\\q_{k+1}}=\matrice{cc}{A & 0\\C_y & I}\matrice{c}{x_{k}\\q_{k}}+\matrice{c}{B\\D_y}u_k,
\label{eq:xq-model}
\end{equation}
where $q$ is the integral of the output $y$. 
Tracking of constant references and rejection of constant disturbances
is achieved by feeding back $q_{k+1}=q_{k}+(y_k-r_k)$ along with $x_k$ in the implementation.
Therefore, assuming that we are given the linear controller
\begin{equation}
    u=-\hat K\matrice{c}{x\\q},
\label{eq:xq-K}
\end{equation}
we can solve the inverse LQR problem for~\eqref{eq:xq-model},~\eqref{eq:xq-K}
as suggested in the previous section, which leads to also weighting the integral state $q_k$ in the stage
cost.

\vspace{0.5em}
\subsubsection*{Input-Output Form}
The controller matching problem in input-output form has been investigated in~\cite{Tran2015} where, however, no guarantee on the existence of a matching controller was proven. In the following, we prove that the above inverse LQR construction can be immediately extended to linear
input-output models and, therefore, all existence guarantees can be extended to the input-output setting.

We consider strictly causal input-output models of the form
\begin{equation}
    \mathcal{A}(z^{-1})y_k=\mathcal{B}(z^{-1})u_k,
    \label{eq:arx}
\end{equation}
where $z^{-1}$ is the backward-shift operator and
\begin{align*}
	\mathcal{A}(z^{-1})=I_{n_y}-\sum_{i=1}^{n_\mathcal{A}}\mathcal{A}_iz^{-i}, && \mathcal{B}(z^{-1})=\sum_{i=1}^{n_\mathcal{B}}\mathcal{B}_iz^{-i}.
\end{align*}
We are given the linear dynamic compensator
\begin{equation}
    \mathcal{\hat C}(z^{-1})u_k=\mathcal{\hat D}(z^{-1})y_k
    \label{eq:IO-controller}
\end{equation}
with
$
	\mathcal{\hat C}(z^{-1}) = I_{n_u}-\sum_{i=1}^{n_\mathcal{C}}\mathcal{\hat C}_i z^{-i}$, $ 
	\mathcal{\hat D}(z^{-1}) = \sum_{i=0}^{n_\mathcal{D}} \mathcal{\hat D}_i z^{-i},
$
and, without loss of generality, $n_\mathcal{C}\leq n_\mathcal{B}$, $n_\mathcal{D}\leq n_\mathcal{A}$. 

Assume that the linear dynamic compensator~\eqref{eq:IO-controller} asymptotically stabilizes~\eqref{eq:arx}. 
In this case the inverse LQR construction described in the previous section can be applied
to the nonminimal state-space realization with state vector
\begin{align}
	\label{eq:io_ss}
	x_k=\left ( y_k, \ \cdots, \ y_{k-n_\mathcal{A}+1}, \ u_{k-1}, \ \cdots, \ u_{k-n_\mathcal{B}+1}\right )\hspace{-2pt},
\end{align}
$x\in\rr^{n_y n_\mathcal{A}+n_u(n_\mathcal{B}-1)}$, by setting
\begin{align*}
	&A\hspace{-2pt}=\hspace{-2pt}\matrice{c|c|c}{\mathcal{A}_1\ \ldots \ \mathcal{A}_{n_\mathcal{A}-1}& \mathcal{A}_{n_\mathcal{A}} & B_2 \ldots\ \mathcal{B}_{n_\mathcal{B}}\\\hline
		I_{(n_\mathcal{A}-1){n_y}} &0& 0\\\hline
		0 & 0 & 0\\\hline
		0   & 0& I_{(n_\mathcal{B}-2)n_u} \ 0}\hspace{-3pt},
	&&B\hspace{-2pt}=\hspace{-2pt}\matrice{c}{\mathcal{B}_1\\\hline 0\\\hline I_{n_u}\\\hline0}\hspace{-3pt}, 
\end{align*}
where, depending on $n_B$, some blocks can have dimension $0$, 
and
\[
	\newcommand{\mysp}{\hspace{6pt}}
    \hat K=-\matrice{c@{\mysp}c@{\mysp}c@{\mysp}c@{\mysp}c@{\mysp}c@{\mysp}c@{\mysp}c@{\mysp}c@{\mysp}c@{\mysp}c@{\mysp}c}{\mathcal{\hat D}_0 & \ldots & \mathcal{\hat D}_{n_\mathcal{D}} & 0 & \ldots & 0 &
    \mathcal{\hat C}_1 & \ldots & \mathcal{\hat C}_{n_\mathcal{C}} & 0 & \ldots & 0}.
\]

The proposed controller matching procedure can then be applied by using the state-space description of the system, provided that $A,B$ are stabilizable and $\hat K$ does stabilize the system.

\vspace{0.5em}
\subsubsection*{Tracking in Input-Output Form}
\label{sec:io_tracking}
Set-point tracking problems can be solved in input-output form by defining the tracking error $e_k=y_k-r_k$ and the input increment $\Delta u_k=u_k-u_{k-1}$, for which the given control law is
\begin{equation}
    \hat C(z^{-1})\Delta u_k=\hat D(z^{-1})e_k.
    \label{eq:IO-controller2}
\end{equation}
In this case, model~\eqref{eq:arx} can be rewritten as
\begin{equation}
    (1-z^{-1})A(z^{-1})y_k=B(z^{-1})\Delta u_k.
    \label{eq:arix}
\end{equation}
For constant references $r_k\equiv r$, by letting $P(z^{-1}):=(1-z^{-1})A(z^{-1})$ we have 
that $P(z^{-1})r_k=0$, which subtracted from~\eqref{eq:arix} gives the tracking error model
\begin{equation}
    P(z^{-1})e_k=B(z^{-1})\Delta u_k.
    \label{eq:arix-err}
\end{equation}
The inverse LQR problem can be now synthesized for model~\eqref{eq:arix-err}
to match the controller~\eqref{eq:IO-controller2} as described above.
This provides a quadratic stage cost that involves $e_k$ and $\Delta u_k$.

\section{MPC Matching Problem}
\label{sec:mpc}
Let us now analyze the case in which linear constraints~\eqref{eq:constraints} 
must be enforced by the controller. This problem is naturally formulated in the Model Predictive Control (MPC) framework, based on solving the following optimal control problem
\begin{subequations}
	\label{eq:mpc}
	\begin{align}
		\min_{w} \ \ & V_\mathrm{f}(x_N) + \sum_{k=0}^{N-1} \ell(x_k,u_k) \label{eq:mpc-cost}\\
		\mathrm{s.t.} \ \ & x_0 = \hat x_0, \\
		& x_{k+1} = Ax_k+Bu_k, && k = 0,\ldots,N-1, \label{eq:mpc-system}\\
		& Cx_k + D u_k +e \leq 0, && k = 0,\ldots,N-1, \label{eq:mpc:pc}\\
		& x_N \in \mathcal{X}_{\mathrm{f}}, \label{eq:mpc:tc}
	\end{align}
\end{subequations}
where $w:=(w_0,\ldots,w_{N-1},x_N)$, $w_k:=(x_k,u_k)$, the stage cost $\ell$ is defined as in~\eqref{eq:LQR-pb}, the terminal cost $V_\mathrm{f}(x_N)$ is quadratic and must be suitably selected together with a corresponding terminal constraint set $\mathcal{X}_{\mathrm{f}}$ to yield recursive feasibility and asymptotic stability~\cite{Rawlings2017}.

Given the current state measurement $\hat x_0$, MPC solves Problem~\eqref{eq:mpc} and applies the first (optimal) control $u_0^\star$ to the system. At the next time step, problem~\eqref{eq:mpc} is solved again using new state measurement in order to close the loop.

Consider the set of states $\mathcal{X}_N:=\{ \, \hat x_0 \, | \, \mu_k^\star(\hat x_0) =0, \nu^\star(\hat x_0)=0 \, \}$, where $\mu_k^\star(\hat x_0)$, $\nu^\star(\hat x_0)$ are the optimal Lagrange multipliers associated with constraints~\eqref{eq:mpc:pc} and~\eqref{eq:mpc:tc}, respectively, when solving~\eqref{eq:mpc}. This is the set of states for which the MPC problem~\eqref{eq:mpc}
and the unconstrained MPC problem~\eqref{eq:mpc-cost}--\eqref{eq:mpc-system} coincide.
The following result is well known in the MPC literature, see, e.g.,~\cite{CM96,SR98,Bemporad2002b}.
\begin{Lemma}
	\label{lem:mpc_equiv_lqr}
	Assume that $\hat x_0 \in \mathcal{X}_N$, $\nabla^2 \ell = H \succ 0$, and $V_\mathrm{f}(x)=x^\top P x$, with $P\succ0$ the solution of the DARE associated with cost $\ell$ along with the corresponding
LQR gain $K$ as in~\eqref{eq:dare}. Then the MPC law~\eqref{eq:mpc} delivers $u_0^\star=-K\hat x_0$.
\end{Lemma}

A set $\mathcal{X}$ is positive invariant for system~\eqref{eq:system} under feedback $u=-\hat Kx$ if $(A-B\hat K)x\in\mathcal{X}$ and $(C-D\hat K)x+e\leq0$, for all $x\in \mathcal{X}$. The maximal positive invariant (MPI) set is the largest positive invariant set, containing all positive invariant sets.
\begin{Lemma}
	If $\mathcal{X}_{\mathrm{f}}$ is selected as the MPI set for the LQR feedback gain $\hat K$, then $\mathcal{X}_N=\mathcal{X}_{\mathrm{f}}$.
\end{Lemma}
\begin{proof}
	By assumption, $\mathcal{X}_{\mathrm{f}}$ is the largest set in which the autonomous system with transition matrix $(A-B\hat K)$ does not violate the path constraints~\eqref{eq:constraints}. Therefore, $\mathcal{X}_{\mathrm{f}}\supseteq \mathcal{X}_N$. Moreover, $\forall \, \hat x_0\in \mathcal{X}_{\mathrm{f}}$ the closed-loop dynamics $u_k=-\hat Kx_k$, $x_{k+1} = Ax_k+Bu_k$, $x_0=\hat x_0$ satisfy $Cx_k+D u_k +e \leq 0$; i.e., $x_k$, $u_k$ are a feasible initial guess for~\eqref{eq:mpc}. Since $\hat K$ is the optimal LQR feedback matrix associated with the stage cost, the guess is also optimal and 
	$\mathcal{X}_{\mathrm{f}}=\mathcal{X}_N$.
\end{proof}

The previous results cover the case in which no constraint is active. With the following Lemma we prove that whenever some constraint is active, the resulting feedback minimizes the deviation from the matched controller.
\begin{Lemma}
	\label{lem:mpc_min_diff}
	Assume that $\ell(x,u)$ is formulated as in~\eqref{eq:stage_cost} and $V_\mathrm{f}(x)=x^\top P x$ where $P$ is the solution to the DARE~\eqref{eq:dare} with $K=\hat K$. Then 
	MPC minimizes the cost
	\begin{align*}
	\sum_{k=0}^{N-1} (u_k+\hat Kx_k)^\top \Gamma (u_k+\hat Kx_k),
	\end{align*}
	with $\Gamma = R+B^\top PB \succ 0.$
\end{Lemma}
\begin{proof}
	The proof follows from Equation~\eqref{eq:qrs_gamma}, which implies 
	\begin{align*}
	&\sum_{k=0}^{N-1} \matr{c}{x_k \\ u_k}^\top H \matr{c}{x_k \\ u_k} + x_N^\top P x_N \\
	&\hspace{6em}= \hat x_0^\top P \hat x_0 + \sum_{k=0}^{N-1} (u_k+\hat Kx_k)^\top \Gamma (u_k+\hat Kx_k).
	\end{align*}
	Since $\hat x_0$ is fixed, the term $\hat x_0^\top P \hat x_0$ is constant and does not influence the optimal solution.
\end{proof}
Note that Lemma~\ref{lem:mpc_min_diff} contains Lemma~\ref{lem:mpc_equiv_lqr} as a special case, since it states that the proposed controller matching procedure guarantees that MPC delivers $u=-\hat K x$ whenever possible, i.e., whenever no constraint becomes active.

We remark that, for $H\succ0$, $P\succ 0$, MPC asymptotically stabilizes system $(A,B)$ to the origin~\cite{Borrelli2017,Rawlings2017,Grune2011}. 
Note that the size of the region of attraction---and feasible domain---of MPC does not decrease with an increasing prediction horizon $N$. In practice one observes that increasing a short prediction horizon $N$ typically leads to a significant increase of the region of attraction.

Finally, the proposed controller matching procedure can easily be coupled with more advanced MPC formulations, e.g., tube-based robust MPC~\cite{Chisci2001,Mayne2005}, which asymptotically stabilizes the closed-loop system to the minimum robust positive invariant set~\cite{Mayne2005}.

The MPC matching procedure is summarized as follows:
\begin{enumerate}
	\item compute $H$ by solving the matching problem~\eqref{eq:sdp_tuning},~\eqref{eq:sdp_tuning_gamma} or~\eqref{eq:sdp_tuning_gamma_opt};
	\item select $P$ as the LQR cost-to-go, obtained automatically in 1); 
	\item compute the terminal set $\mathcal{X}_\mathrm{f}$ as the MPI set for feedback $K$. 
\end{enumerate}

\vspace{0.5em}
\subsubsection*{Nonlinear MPC}
We consider now the case of a nonlinear system 
\begin{align*}
	x_{k+1} = f(x_k,u_k).
\end{align*}
One can linearize the system around a steady state $x_\mathrm{s},u_\mathrm{s}$ to obtain
\begin{align*}
	A = \nabla_x f(x_\mathrm{s},u_\mathrm{s})^\top, && B = \nabla_u f(x_\mathrm{s},u_\mathrm{s})^\top, \\
	C = \nabla_x h(x_\mathrm{s},u_\mathrm{s})^\top, && D = \nabla_u h(x_\mathrm{s},u_\mathrm{s})^\top, && e=h(x_\mathrm{s},u_\mathrm{s}),
\end{align*}
and use a linear controller to locally stabilize the nonlinear system. Then, the controller matching strategy can be deployed as described before to define a matching linear MPC problem.

In case one is interested in further improving performance by using a nonlinear model within MPC, Nonlinear MPC (NMPC) can be formulated as follows~\cite{Rawlings2017,Grune2011}
\begin{subequations}
	\label{eq:nmpc}
	\begin{align}
	\min_{w} \ \ & V_\mathrm{f}(x_N) + \sum_{k=0}^{N-1} \ell(x_k,u_k) \label{eq:nmpc-cost}\\
	\mathrm{s.t.} \ \ & x_0 = \hat x_0, \\
	& x_{k+1} = f(x_k,u_k), && k = 0,\ldots,N-1, \label{eq:nmpc-system}\\
	& h(x_k,u_k) \leq 0, && k = 0,\ldots,N-1, \label{eq:nmpc:pc}\\
	& x_N \in \mathcal{X}_{\mathrm{f}}. \label{eq:nmpc:tc}
	\end{align}
\end{subequations}

\begin{Lemma}
	Assume that the stage cost is selected as the solution to the controller matching problem~\eqref{eq:sdp_tuning}, \eqref{eq:sdp_tuning_gamma}, or~\eqref{eq:sdp_tuning_gamma_opt} for the system linearization computed at $x_\mathrm{s},u_\mathrm{s}$. Assume further that $h(x_\mathrm{s},u_\mathrm{s})<0$ and the terminal cost is selected as $V_\mathrm{f}(x)=x^\top P x$, with $P\succ0$ the solution of the DARE associated with cost $\ell$ and the system linearized at $x_\mathrm{s},u_\mathrm{s}$. Then, the NMPC feedback $u^*_0(\hat x)$ satisfies 
	\begin{align*}
		\| u^*_0(\hat x) + \bar K \hat x \| = O(\| \hat x - x_\mathrm{s} \|^2).
	\end{align*}
\end{Lemma}
\begin{proof}
	By relying on the results derived in~\cite{Guddat1990,Diehl2001} we note that, by construction, the feedback control law $u^\mathrm{NMPC}(\hat x_0)$ yielded by the NMPC formulation~\eqref{eq:nmpc} and the one yielded by the linear MPC formulation~\eqref{eq:mpc}, i.e., $u^\mathrm{MPC}(\hat x_0)$, satisfy
	\begin{align*}
		\nabla_{\hat x_0} u^\mathrm{NMPC}(x_\mathrm{s}) = \nabla_{\hat x_0} u^\mathrm{MPC}(x_\mathrm{s}).
	\end{align*}
	A more detailed proof can be found in~\cite[Appendix B]{Zanon2016b}.
\end{proof}

\section{Observers and Moving Horizon Estimation}
\label{sec:observer}

In this section, we discuss how the proposed controller matching procedure can be applied to the state estimation problem. This allows one to interpret any linear observer as a Kalman filter and to formulate Moving Horizon Estimation (MHE) which locally behaves like the linear observer, and handles constraints and nonlinear dynamics. 

Note that, while MHE is often formulated using the Kalman filter for tuning, MHE observers can be tuned using other criteria, e.g., $\mathcal{H}_\infty$~\cite{Tirado2018},

where, due to the computational complexity of solving a minimax problem, the problem is solved only approximately. We need to stress that with our tuning procedure the $\mathcal{H}_\infty$-tuned MHE problem can be solved exactly and efficiently for linear systems, since one needs to solve a convex QP instead of a minimax problem.

Consider the following linear system 
\begin{align*}
	x_+ &= Ax + w, & y &= C_yx + v,
\end{align*}
where $w$ and $v$ denote process and measurement noise, respectively.
We write the one-step-ahead estimation problem at time $n$ as
	\begin{align}
		x_-^*,x_+^* = \arg\min_{x_-,x_+} \ & \matr{c}{Ax_--x_+ \\ C_yx_--y}^{\hspace{-1pt}\top} \hspace{-1pt} H^{-1} \matr{c}{Ax_--x_+ \\ C_yx_--y} \nonumber \\ 
		&+ (x_--\hat x)^\top P^{-1}(x_--\hat x), \label{eq:kalman_generalized}
	\end{align}
where the estimation error covariance is $P=\E[(x-\hat x)(x-\hat x)^\top]$ and the measurement and process noise covariance is
\begin{align*}
	H=\matr{ll}{Q & S^\top \\ S &R}=\E\left [\matr{c}{w \\ v}\matr{c}{w \\ v}^\top\right ],
\end{align*}
where in Kalman filtering one often assumes $S=0$. The optimal state estimate is then $\hat x_+=x_+^*$. Note that we used a compact notation for the state estimates, which are usually denoted as $\hat x = x_{n|n-1}$; $x_-^*=x_{n|n}$; and $\hat x_+=x_{n+1|n}$ to explicitly state which information they use to predict the state at which time.

The Kalman filter covariance update is given by the DARE~\cite{Deshpande2017}
\begin{subequations}
	\label{eq:kalman_covariance}
	\begin{align}
		P_+ &= APA^\top + Q - L(S^\top + C_yPA^\top), \\
		L &= (S+APC_y^\top)(R+C_yPC_y^\top)^{-1}, 
	\end{align}
\end{subequations}
where $P_+=P$ at steady state. The Kalman filter state estimate is
\begin{align}
	\label{eq:kalman_estimate}
	\hat x_+ = A \hat x - L(C_y\hat x-y).
\end{align}

\begin{Lemma}
	The estimation problem~\eqref{eq:kalman_generalized} coincides with a Kalman filtering problem and delivers state estimate~\eqref{eq:kalman_estimate} with feedback gain and covariance update given by~\eqref{eq:kalman_covariance}.
\end{Lemma}
\begin{proof}
We define
\begin{align*}
	\matr{ll}{\tilde Q & \tilde S^\top \\ \tilde S &\tilde R} = \matr{ll}{Q & S^\top \\ S &R}^{-1} = H^{-1}.
\end{align*}
The matrix inversion lemma and the Schur complement yield 
\begin{align}
	\label{eq:matrix_inverse_formulae}
	\tilde Q^{-1} \tilde S^\top = -S^\top R^{-1}, && \tilde R- \tilde S\tilde Q^{-1} \tilde S^\top = R^{-1}.
\end{align}
Then, the first-order necessary conditions for optimality of the problem in~\eqref{eq:kalman_generalized} read
\begin{align*}
	0&=-\tilde S^\top (C_yx_-^*-y) - \tilde Q (Ax_-^*-x_+^*), \\
	0&=P^{-1}(x_-^*-\hat x) +C_y^\top \tilde R (C_yx_-^*-y) + C_y^\top \tilde S (Ax_-^*-x_+^*) \\
	&\hspace{7em}+ A^\top \tilde S^\top (C_yx_-^*-y) +A^\top \tilde Q (Ax_-^*-x_+^*).
\end{align*}
From the first condition we get
\begin{align*}
	x_+^* = Ax_-^* + \tilde Q^{-1}\tilde S^\top (C_yx_-^*-y).
\end{align*}
By inserting this in the second condition and using~\eqref{eq:matrix_inverse_formulae}, we have
$
	(P^{-1} + C_y^\top R^{-1}C_y)x_-^* = P^{-1}\hat x + C_y^\top R^{-1} y.
$
By using the matrix inversion lemma 
one can derive
\begin{align*}
	x_-^* = \hat x - PC_y^\top (R+C_y P C_y^\top)^{-1}(C_y\hat x-y).
\end{align*}
Then, we can conclude that
\begin{align}
	\hat x_+=x_+^* &= Ax_-^* - S^\top R^{-1} (C_yx_-^*-y) 
	= A \hat x - L(C_y\hat x-y), \nonumber \\
	L&=(S^{\top} + APC_y^\top )(R+C_yPC_y^\top)^{-1}.\label{eq:kalman_gain}
\end{align}
Let us denote the estimation error as $e=\hat x - x$, which entails
\begin{align*}
	e_+&=Ae - L(C_y\hat x-y)-w 
	= (A-LC_y)e -w + Lv,
\end{align*}
and remind that $\E[ee^\top]=P$, $\E[ew^\top]=0$, $\E[ev^\top]=0$, $\E[ww^\top]=Q$, $\E[vv^\top]=R$, $\E[wv^\top]=S.$

The covariance $P_+:=\mathbb{C}[e_+]$ of the estimation error update is:
\begin{align*}
P_+ 
&= (A-LC_y)P(A-LC_y)^\top + LRL^\top + Q - SL^\top - L S^\top\\
&= APA^\top + Q - L(C_yPA^\top +S^\top),
\end{align*}
where we used~\eqref{eq:kalman_gain} to derive the last equality and obtain~\eqref{eq:kalman_covariance}.
\end{proof}

We proved that the one-step-ahead estimation problem~\eqref{eq:kalman_generalized} coincides with a Kalman filter whose feedback and covariance propagation are given by DARE~\eqref{eq:kalman_covariance}. Note that~\eqref{eq:kalman_covariance} coincides with~\eqref{eq:dare}, if $(A,B)$ is replaced by $(A^\top,C_y^\top)$, such that $L=K^\top$. Therefore, the controller matching procedure also applies to linear observers and can be used to formulate MHE problems which match any desired linear observer yielding asymptotically stable state-estimation errors. 

\section{Practical Implementation}
\label{sec:simulations}

In this section we demonstrate the theory with some examples and show how the matching technique can be implemented in practice. 

\vspace{0.5em}
\subsubsection*{Tuning Matrix $\Gamma$}
	\label{ex:constrained}
	We illustrate how different choices of $\Gamma$ can influence the optimal solution in the presence of active constraints. We remark that, by construction, whenever no constraint is active any $\Gamma \succ 0$ delivers the same feedback.
	Consider the discrete-time linear system defined by
	\begin{align*}
	A = -0.8,  && B = \tightMatr{ccc}{0.1 & 0.1 & 0.1},
	\end{align*}
subject to the constraint $x\leq 0.7$.
Consider the stabilizing gain
\[
    \hat K = \tightMatr{ccc}{0.5 & 0.5 & 0.2}^\top.
\]
We want to synthesize an MPC controller with prediction horizon $N=1$, terminal LQR cost, terminal constraint set $\mathcal{X}_f=\{x|x\leq0.7\}$,
and such that the MPC law coincides with $\hat K$ when constraints are inactive in the
MPC problem.
Consider the two weighting matrices $\Gamma_1 = I$ and $\Gamma_2 = \mathrm{diag}(\tightMatr{ccc}{1 & 100 & 1})$. Moreover, consider the tuning matrix obtained by solving the direct formulation~\eqref{eq:sdp_tuning}: 
	\begin{align*}
	H = \matr{rrrr}{
		1.3128 &   0.6917  &  0.7088 &   0.4775   \\
		0.6917 &   1.1610  & -0.1849  &  0.1173    \\
		0.7088 &  -0.1849 &   1.2435 &  -0.0036 \\
		0.4775  &  0.1173  & -0.0036  &  1.2021   \\
	}.
	\end{align*}
	For $\hat x_0=-1$, we have $(A-B\hat K)\hat x_0=0.92>0.7$: MPC deviates from the desired controller to satisfy the constraint. We obtain the following controls (the subscript denotes the used weighting matrix):
	\begin{equation*}
	u_{\Gamma_1} = \tightMatr{r}{-0.2333 \\ -0.2333 \\ -0.5333 }, \ \, u_{\Gamma_2} = \tightMatr{r}{-0.5945 \\ \phantom{-}0.4891 \\ -0.8945 }, \ \, u_H = \tightMatr{r}{ -0.2849 \\ -0.2923 \\ -0.4228}.
	\end{equation*}
	The tuning role of matrix $\Gamma$ is best understood by considering the cost in form~\eqref{eq:indef_cost}: $\Gamma$ does not penalize the use of the controls themselves, but rather their deviation from $-\hat K \hat x_0$: 
	\begin{align*}
	|u_{\Gamma_1}+\hat K\hat x_0| = \tightMatr{r}{0.7333 \\ 0.7333 \\ 0.7333 }, \qquad |u_{\Gamma_2}+\hat K\hat x_0| = \tightMatr{r}{1.0945 \\ 0.0109 \\ 1.0945 },
	\end{align*}
	i.e., for $\Gamma_2$, the second control is closer to its reference value $0.5$ than for $\Gamma_1$, but larger in magnitude. Since tuning matrix $H$ is obtained through the direct formulation~\eqref{eq:sdp_tuning}, no choice can be made on how the controls are selected in the presence of active constraints. 

\vspace{0.5em}
\subsubsection*{PID and Input-Output Form}
	\label{ex:pid}
	
	\begin{figure}
		\begin{center}
			\includegraphics[width=\linewidth]{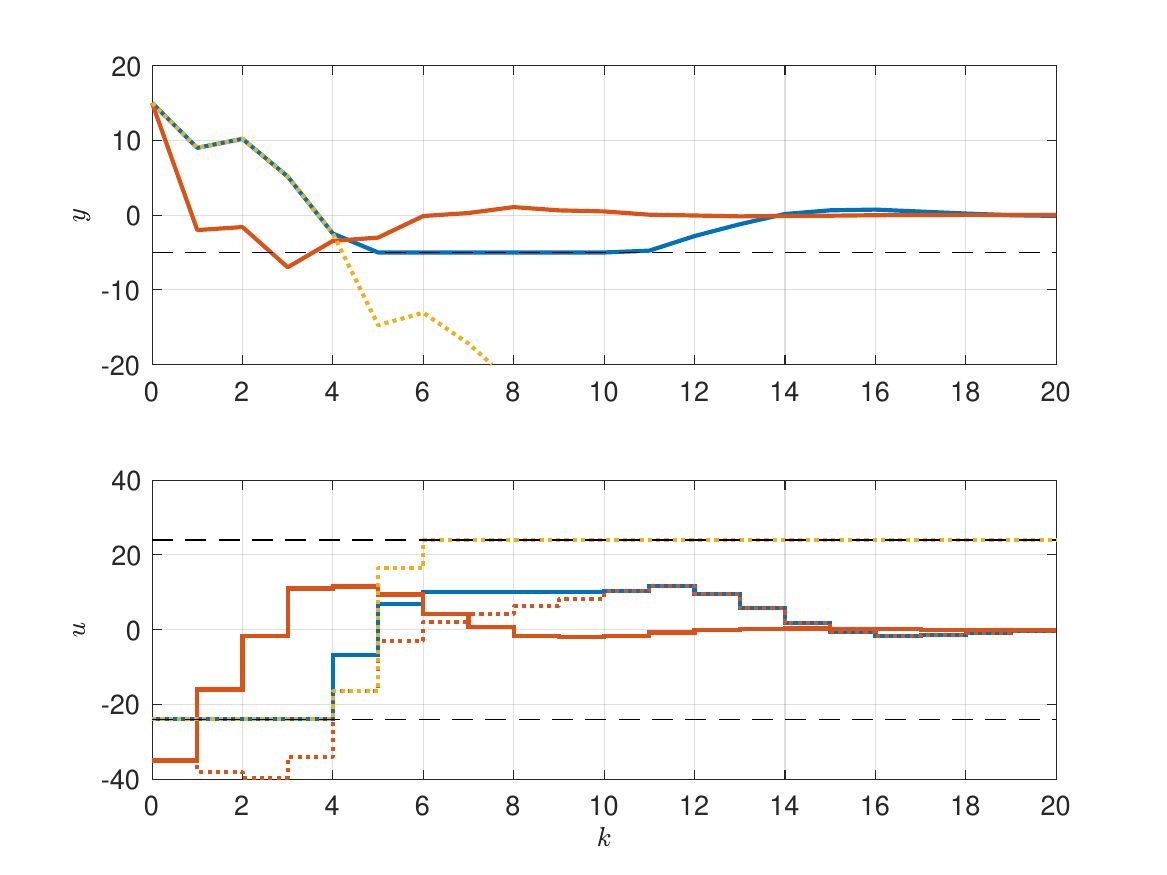}
			\caption{Closed-loop simulations for the PID example. Top plot: MPC (blue line), PID (red line), PID with saturated control (yellow dotted line). Bottom plot $u=-\hat K x$, with $x$ form the MPC closed-loop trajectory (dotted red line). Constraints are displayed in dashed black line.}
			\label{fig:pid}
		\end{center}
	\end{figure}

	Consider the linear system in input-output form from~\cite{DiCairano2010}
	\begin{align*}
		y_k = 1.8y_{k-1} + 1.2y_{k-2} + u_{k-1},
	\end{align*}
	with sampling time $t_\mathrm{s}=2$ and subject to constraints
		$-24\leq u\leq 24$, and $ y \geq -5$.
	We want to match the PID controller
	\begin{align*}
		u^\mathrm{PID}_k &= -\left ( K_\mathrm{i} \, y^\mathrm{i}_k + K_\mathrm{p} \, y_k + \frac{K_\mathrm{d}}{t_\mathrm{s}}\left (y_k-y_{k-1}\right )  \right ), \\
		y^\mathrm{i}_k &= y^\mathrm{i}_{k-1} + t_\mathrm{s}y_{k},
	\end{align*}
	with $K_\mathrm{i}=0.248$, $K_\mathrm{p}=0.752$, $K_\mathrm{d}=2.237$. We write the system dynamics in the state-space representation $x_{k+1}=Ax_k+Bu_k$ where
	\begin{align*}
		x_k \hspace{-1pt}=\hspace{-1pt} \matr{c}{y_{k-1} \\ y_{k-2} \\ y^\mathrm{i}_{k-1} \\ u_{k-1}}\hspace{-1pt}, && 
		A \hspace{-1pt}=\hspace{-1pt} \matr{cccc}{1.8 & 1.2 & 0 & 1 \\ 1 & 0 & 0 & 0 \\ 3.6 & 2.4 & 1 & 2 \\ 0  & 0 & 0 & 0}\hspace{-1pt}, &&
		B \hspace{-1pt}=\hspace{-1pt} \matr{c}{0 \\ 0 \\ 0 \\ 1}\hspace{-1pt}.
	\end{align*}
	Then, the PID becomes $u=-\hat Kx$, with 
	\begin{align*}
		K_\mathrm{pid}&=K_\mathrm{p} + K_\mathrm{i}t_\mathrm{s} + K_\mathrm{d}/t_\mathrm{s}, \\
		\hat K&=\matr{cccc}{ K_\mathrm{d}/t_\mathrm{s} + 1.8 K_\mathrm{pid} & 1.2 K_\mathrm{pid} & K_\mathrm{i} & K_\mathrm{pid} }\\
		&=\matr{cccc}{5.3782  &  2.8398  &  0.2480  &  2.3665}.
	\end{align*}
	In~\cite{DiCairano2010} an LQR with dense $Q$ provided an exact match. We are able to reproduce the same result by either adding the constraint that $S=0$ or minimizing, e.g., $\|S\|_1$. By minimizing the condition number of $H$, we obtain $\kappa(H)\approx 1.7$, as opposed to $\kappa(H)\approx 6.6$ found by~\cite{DiCairano2010}. In this case, there is no clear advantage in minimizing the condition number, since $6.6$ is so low that it does not cause numerical issues. If we minimize the condition number of $\boldsymbol{H}=\mathrm{blkdiag}(H,P)$, we obtain $\kappa(\boldsymbol{H})\approx 158.8$ with $S=0$ and $\kappa(\boldsymbol{H})\approx 149.2$ with $S$ free.
	
	We plot the control and output closed-loop trajectories in Figure~\ref{fig:pid}. MPC respects the constraints and, as soon as the output enters the region in which no output nor input constraints would be active under the feedback $\hat K$, MPC and PID deliver the same control (blue and dotted red lines). The PID controller violates both the input and output constraints (red line). By saturating the PID input to satisfy the input constraint, the output is not stabilized (dotted yellow line).

	Consider now the desired feedback law
	\begin{align*}
		\hat K&=\matr{cccc}{4  &  2  &  0.15  &  1.6}.
	\end{align*}
	In this case, with $S=0$ there exists no LQR matching the feedback $\hat K$, though it is stabilizing. By allowing $S\neq0$, one is able to compute $H\succ0$ such that $K=\hat K$. The condition number is $\kappa(H)\approx30.5$. 

\vspace{0.5em}
\subsubsection*{From Continuous to Discrete Time and Anti Wind-Up}
	\label{ex:data_driven_cstr}

	Consider the nonlinear continuously-stirred tank reactor (CSTR) with dynamics~\cite{Seborg2003}
		\begin{align*}
			\dot T &= \frac{q}{V}(T_\mathrm{f}-T) + \frac{H_{\mathrm{AB}}}{\rho C_\mathrm{p}} K_0 e^{\frac{E}{RT}}C_\mathrm{A} + \frac{UA}{V \rho  C_\mathrm{p} }(T_\mathrm{c}-T), \\
			\dot C_\mathrm{A} &= \frac{q}{V}(C_\mathrm{Af}-C_\mathrm{A}) - K_0 e^{\frac{E}{RT}}C_\mathrm{A}, 
		\end{align*}
	with state $x=(T,C_\mathrm{A})$ (temperature and concentration of $A$); control $u=T_\mathrm{c}$ (temperature of the cooling jacket); and output $y=T$. The parameters are: volume $V=100 \ \mathrm{m}^3$, density of the A-B mixture $\rho = 1000 \ \mathrm{kg}/\mathrm{m}^3$, reaction heat $H_\mathrm{AB} = 5\cdot 10^4 \ \mathrm{J}/\mathrm{mol}$, activation energy over the universal gas constant $E/R = 8750 \ \mathrm{J}/\mathrm{mol \ K}$, time constant $K_0=7.2\cdot 10^{10} \ 1/\mathrm{s}$ and the heat transfer coefficient times the area $UA=5\cdot 10^4 \ \mathrm{W}/\mathrm{K}$. The system is subject to the uncontrolled volumetric flowrate $q = 1 \pm 0.1 \ \mathrm{m}^3/\mathrm{s}$, feed concentration $C_\mathrm{Af}=1 \pm 0.1 \ \mathrm{mol}/\mathrm{m}^3$, and feed temperature $T_\mathrm{f} = 350 \pm 10 \ \mathrm{K}$.
	
	The system is already controlled by a PI controller with proportional gain $K_\mathrm{p}=0.5$, integral gain $K_\mathrm{i}=5$ and an anti-windup gain $K_\mathrm{aw}=1$ such that the integral term is given by
	\begin{align*}
		\dot I_e = e + K_\mathrm{aw} \min(\max(K_\mathrm{p}e + K_\mathrm{i}I_e, \, u_\mathrm{lb}), \,u_\mathrm{ub}),
	\end{align*}
	with $e=T_\mathrm{ref}-T$; $T_\mathrm{ref}$ the reference setpoint and $u_\mathrm{lb} = 250 \ \mathrm{K}$, $u_\mathrm{ub} = 350 \ \mathrm{K}$ the saturation bounds on the control signal.
	
	Consider 
	the setpoint $x_\mathrm{s}=(300,0.39,59.72)$ $u_\mathrm{s}=298.59$ with output reference $r_\mathrm{s}=300$. We write the system dynamics in closed-loop with the PI controller as $\dot x = f^\mathrm{c}_\mathrm{PI}(x,r)$ and linearize them at $x_\mathrm{s}$, $r_\mathrm{s}$ to obtain the continuous- and discrete-time matrices
	\begin{align*}
		A^\mathrm{c}_\mathrm{PI} = \frac{\partial f^\mathrm{c}_\mathrm{PI}}{\partial x}, && A_\mathrm{PI}&=e^{A^\mathrm{c}_\mathrm{PI} t_\mathrm{s}},\\
		B^\mathrm{c}_{r,\mathrm{PI}} = \frac{\partial f^\mathrm{c}_\mathrm{PI}}{\partial r}, && B_{r,\mathrm{PI}} &= \int_0^{t_\mathrm{s}} e^{A^\mathrm{c}_\mathrm{PI} t} B^\mathrm{c}_{r,\mathrm{PI}} \mathrm{d}t,
	\end{align*}
	for a sampling time $t_\mathrm{s}$. 
	
	We apply the same procedure to the open-loop dynamics $\dot x = f^\mathrm{c}(x,u)$, linearized at $x_\mathrm{s}$, $u_\mathrm{s}$ to get
	\begin{align*}
		A^\mathrm{c} = \frac{\partial f^\mathrm{c}}{\partial x}, && B^\mathrm{c} = \frac{\partial f^\mathrm{c}}{\partial u}, && B^\mathrm{c}_r = \frac{\partial f^\mathrm{c}}{\partial r},
	\end{align*}
	and the corresponding discrete-time linearized system
	$
		\Delta x_{k+1} = A\Delta x_k + B\Delta u_k + B_r \Delta r_k.
	$
	From the the continuous-time PI feedback, we compute the corresponding discrete-time feedback matrix $\hat K$ by pole placement, i.e., by imposing:
	$
		\mathrm{eig}(A-B\hat K) = \mathrm{eig}(A_\mathrm{PI}).
	$.
	
	We compute the reference for MPC as the steady-state $\Delta x^\mathrm{r}$, $\Delta u^\mathrm{r}$ associated with a given $\Delta r$:
	\begin{align*}
		\Delta x^\mathrm{r} &= (A_\mathrm{PI})^{-1}B_{r,\mathrm{PI}}\Delta r, \\
		\Delta u^\mathrm{r} &= \arg\min_{\Delta u} \ \| A\Delta x^\mathrm{r} +B\Delta u + B_\mathrm{r}\Delta r \|,
	\end{align*}
	where, by construction, $A \Delta x^\mathrm{r} +B\Delta u^\mathrm{r} + B_\mathrm{r}\Delta r=0$.
	
	Since the integral state is not a state of the system but a state of the controller, its time propagation is given by the MPC prediction. We introduce an anti-windup mechanism by adding the term 
	\begin{equation*}
	f_\mathrm{aw}(\Delta x,\Delta u) := t_\mathrm{s}K_\mathrm{aw} \left (\Delta u + u_\mathrm{s} - \Delta u^\mathrm{r} - \bar K (\Delta x + x_\mathrm{s}-\Delta x^\mathrm{r}) \right )
	\end{equation*}
    to the dynamics of the integral state, where $K_\mathrm{aw}$ is the PI anti-windup gain. 
The MPC formulation then reads
\begin{align*}
	\min_{\Delta x,\Delta u} \ \ &\sum_{k=0}^{N-1} \matr{c}{\Delta x_k - \Delta x^\mathrm{r}_k\\ \Delta u_k - \Delta u^\mathrm{r}_k}^\top H \matr{c}{\Delta x_k - \Delta x^\mathrm{r}_k\\ \Delta u_k - \Delta u^\mathrm{r}_k} \\
	&\hspace{6em}+ (\Delta x_N - \Delta x^\mathrm{r}_N)^\top P (\Delta x_N - \Delta x^\mathrm{r}_N) \\
	\mathrm{s.t.} \ \ & \Delta x_0 = \hat x - x_\mathrm{s},\\
	& \Delta x_{k+1} \hspace{-1pt}=\hspace{-1pt} A \Delta x_k \hspace{-1pt}+\hspace{-1pt} B \Delta u_k \hspace{-1pt}+\hspace{-1pt} B_\mathrm{r} \Delta r_k \hspace{-1pt}+\hspace{-1pt} \tightMatr{c}{0\\0\\f_\mathrm{aw}(\Delta x,\Delta u)},\\
	& C\Delta x_k + D \Delta u_k \leq e.
\end{align*}

A nonlinear MPC can also be formulated, where the dynamics are
\begin{align*}
	\Delta x_{k+1} = f( \Delta x_k , \Delta u_k,\Delta r_k) + \tightMatr{c}{0\\0\\f_\mathrm{aw}(\Delta x,\Delta u)}.
\end{align*}

We compare in simulations the PI controller with the linear MPC controller on a reference step change. Additionally, we introduce a constraint on the maximum temperature and simulate both the MPC (MPCx) and NMPC controllers. The results are displayed in Figure~\ref{fig:lin_mpc_nmpc}: the PI controller violates the temperature constraint; linear MPC is also violating it due to the linearization error which causes an inaccurate prediction; NMPC does satisfy this constraint and stabilizes the system to the desired output. If the temperature constraint is removed, NMPC has a smaller overshoot for the considered step of $30$ degrees. For a step of $10$ degrees all controllers are qualitatively the same.

\begin{figure}
	\begin{center}
		\includegraphics[width=\linewidth]{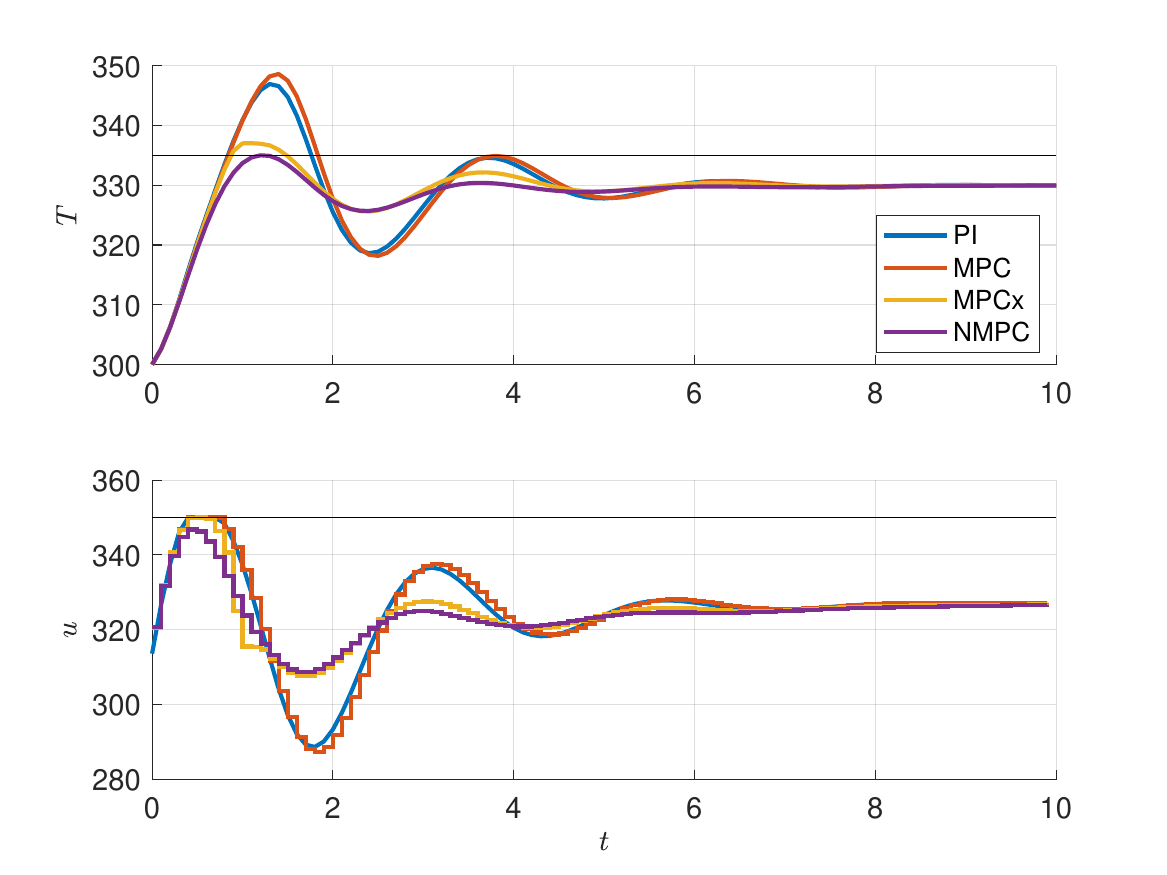}
		\caption{Temperature and input profiles for a closed-loop simulation using the PI controller (PI), the MPC controller (MPC), the MPC controller with a constraint on the maximum temperature (MPCx) and nonlinear MPC with the same constraint (NMPC).}
		\label{fig:lin_mpc_nmpc}
	\end{center}
\end{figure}

\vspace{0.5em}
\subsubsection*{$\mathcal{H}_\infty$ Moving Horizon Estimation}

With the following example, we detail how a robust MHE can be formulated, based on classical results for robust linear observers. Given the full equivalence with control problems, we remark that this also applies to robust tuning of MPC controllers.

Consider the system defined by 
\begin{align*}
	x_+ &= \left ( A + \matr{cc}{0 & 0 \\ -2x_2& 0.1x_2} \right )x+Bw, \\ y&=C_yx+v,\\
	A &= \matr{rr}{0.93& 0.09 \\ -0.61& 0.92}, && \hspace{-5em} B = \matr{ll}{0.01& 0.01\\
		0.003& 0.12}, \\
	C_y &= \matr{cc}{1 & 0},
\end{align*}
with process noise covariance $W=\mathrm{diag}(\matr{cc}{10 & 10})$ and measurement noise covariance $V=0.01$. By neglecting the nonlinear term, one can design both a Kalman filter and an $\mathcal{H}_\infty$ observer, which solves~\cite{Simon2006a}:
\begin{align*}
	\Sigma &= ( I - PG^\top G + P C_y^\top V^{-1}C_y )^{-1}P, \\ 
	L &= A\Sigma C_y^\top V^{-1}, \\
	P &= A\Sigma A^\top + BWB^\top,
\end{align*}
where we select tuning parameter $G=\gamma \, \mathrm{diag}(\matr{cc}{0.1 & 1})$, and $\gamma$ is a scalar to be maximized. 

For our example, the $\mathcal{H}_\infty$ observer is obtained for $\gamma\approx1.3438$. The two observers yield feedback 
\begin{align*}
L^\mathrm{Kalman}=\matr{c}{0.6866 \\ 1.5202}, && 	L^{\mathcal{H}_\infty}=\matr{c}{1.4391 \\ 4.5947}.
\end{align*}
The tuning procedure yields 
\begin{align*}
H_{\mathcal{H}_\infty}^{-1} = \matr{rrr}{     0.9451  & -0.2260 &  -0.0239 \\ -0.2260 & 0.0693 & -0.0985 \\ -0.0239 & -0.0985 & 0.9896},
\end{align*}
as weighting matrix for the $\mathcal{H}_\infty$-tuned MHE; while the Kalman MHE formulation uses the inverse of the noise covariance, i.e., $H_{\mathrm{Kalman}}^{-1} = \mathrm{diag}(\matr{ccc}{0.1&0.1&100})$.

We assume that we have knowledge about the fact that $w\geq0$. We include this information by using the proposed tuning procedure to design a cost for the Kalman filter such that it yields the $\mathcal{H}_\infty$ observer and then use the obtained cost within a linear (MHE) framework.  

The simulation results are displayed in Figure~\ref{fig:h_inf}, where one can see that the two MHE perform similarly. The root mean square (RMS) error obtained with the $\mathcal{H}_\infty$-tuned MHE is $176.9$, while for a standard MHE formulation we obtain an RMS error of $214.9$. Note that the $\mathcal{H}_\infty$ and Kalman filter have an RMS error of $208.4$ and $215.5$ respectively.

\begin{figure}
	\begin{center}
		\includegraphics[width=\linewidth]{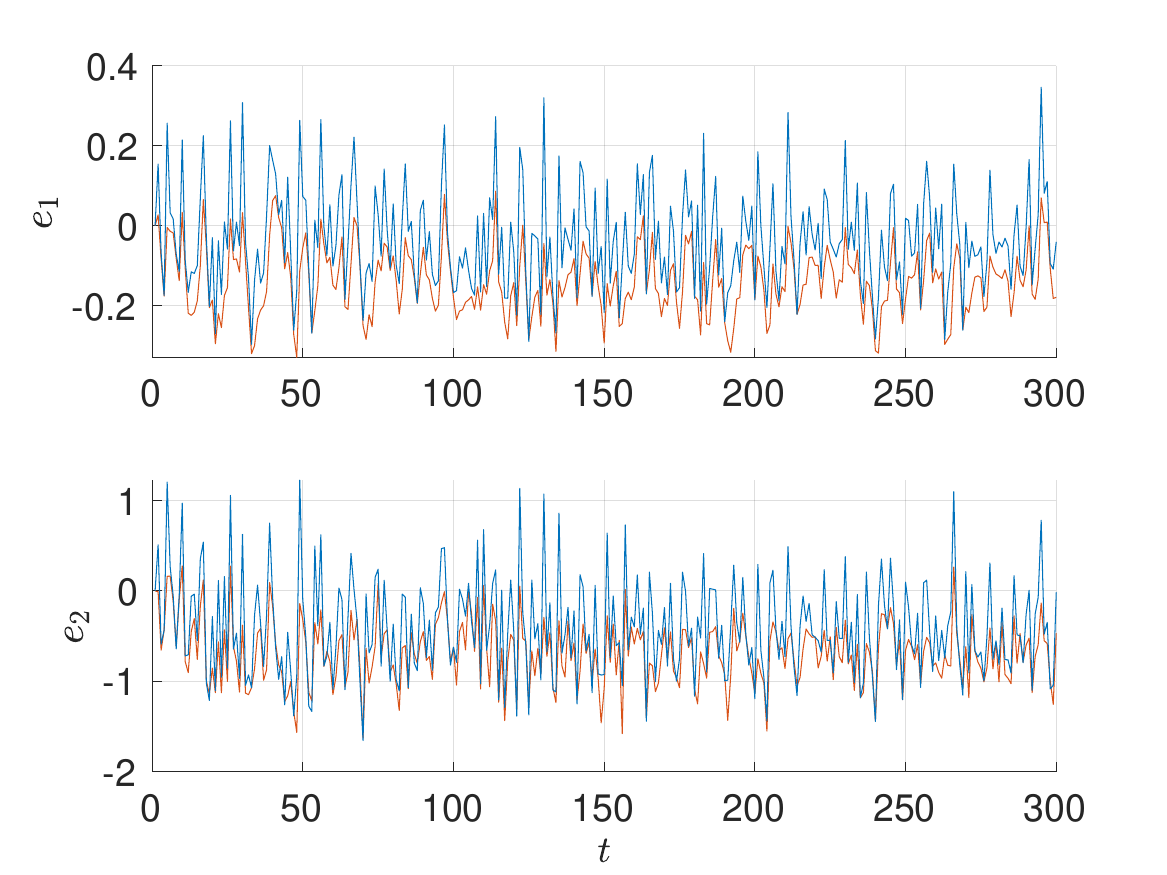}
		\caption{State estimation error: tuned MHE (blue), standard MHE (red).}
		\label{fig:h_inf}
	\end{center}
\end{figure}

\section{Conclusions}
\label{sec:conclusions}

This paper analyzed how to design a LQR/MPC cost function that results in a prescribed linear control law. We have proven that the problem can be solved exactly for all stabilizing linear controllers, both in state-space and input-output form, and provided three approaches to compute the desired cost by solving a convex SDP. The results also extend to linear observers, which can be matched by a Kalman filter or MHE. 

\bibliographystyle{IEEEtran}
\bibliography{bibliography}

\end{document}